%% file: workshop_paper.tex
\documentclass[sigconf, nonacm]{acmart}

\newcommand\vldbavailabilityurl{}   % leave empty or add your artifacts URL

\usepackage{graphicx} % Required for inserting images

\usepackage{amsmath,amsfonts,amsthm,amssymb}
\usepackage{thmtools}
\usepackage{url}
\usepackage{hyperref}
\usepackage{cleveref}
\usepackage{algorithm}
\usepackage{algpseudocode}
\algrenewcommand\algorithmicrequire{\textbf{Input:}}
\algrenewcommand\algorithmicensure{\textbf{Output:}}
\newcommand{\algrule}[1][.2pt]{\par\vskip.2\baselineskip\hrule height #1\par\vskip.2\baselineskip}
\usepackage{tikz}
\usepackage{bbm}
\usepackage{stfloats}
\usepackage{booktabs}   % \toprule, \midrule, \cmidrule, \bottomrule
\usepackage{multirow}   % \multirow for the spanning Dataset / Recall@k cells
\usepackage{array}

\input{bib_macros}

\usetikzlibrary{patterns, arrows.meta}

\input{math_macros.tex}

\title{Almost Navigable Graphs}

\author{Pratyush Avi}
\affiliation{%
  \institution{New York University}
  \city{}
  \country{}
}
\email{pratyushavi@nyu.edu}

\author{Christopher Musco}
\affiliation{%
  \institution{New York University}
  \city{}
  \country{}
}
\email{cmusco@nyu.edu}

\begin{document}

    \begin{abstract}
    Graph-based methods like HNSW, DiskANN, NSG, and others have become an increasingly popular choice for implementing approximate nearest neighbor search (ANNS) in Vector Databases (VecDBs). The success of these methods has motivated the study of how to best construct a search graph for a given dataset. To that end, \emph{navigability} has been identified as a desirable graph property which ensures good ANNS performance when combined with greedy search. 
    
    However, for a dataset with $n$ vectors, the sparsest navigable graph requires $O(n\sqrt{n})$ edges in the worst-case, and we show empirically that, for typical billion node datasets, 100s of edges are needed per node. This leads to slow search and high memory requirements. Moreover, under standard complexity theoretical assumptions, it was recently established that constructing a sparse navigable graph requires $\Omega(n^{2-\epsilon})$ time, which is prohibitive for large datasets.

    We address these concerns by introducing a relaxed notation of navigability called ``$\gamma$-almost navigability'' for any $\gamma \in [0,1]$, with $\gamma = 1$ corresponding to full navigability. We prove that any dataset (under any distance) admits a $\gamma$-almost navigable graph with just $O\left(\frac{n}{1-\gamma}\right)$ edges, linear in the dataset size. We present a randomized algorithm for constructing such a graph in near-linear time. 

    While we prove that $\gamma$-almost navigability sacrifices the worst-case search guarantees enjoyed by navigability, we  show empirically that greedy beam search still performs well in such graphs when $\gamma < 1$. Indeed, we obtain improved recall-runtime tradeoffs on a variety of datasets compared to fully navigable graphs. Moreover, our graphs are more space efficient, with degree typically less than half that of a fully navigable graph for comparable performance.
    \end{abstract}

    \maketitle

    \section{Introduction}
    Approximate nearest neighbor search (ANNS) is the central algorithmic problem underlying vector databases (VecDBs). Formally, we are given a set of points $P = \{p_1, \dots, p_n\} \subset \R^m$ and a distance function $d:\R^m \times \R^m \to \R$. The goal is to preprocess the points into a data structure so that, given any query vector $q \in \R^m$, we can efficiently find the $k$-nearest neighbors of $q$ in $P$, or at least $k$ points that are near-minimizers of $d(x,q)$ over $x\in P$. 
    
    For standard distance functions, the complexity of classical approaches for the ANN problem, like k-d trees, scales exponentially in the data dimension, $m$. This has motivated a wide variety of alternative approaches specifically targeted at high-dimensional datasets, including locality-sensitive hashing \cite{IndykMotwani:1998,AndoniIndyk:2008,AndoniIndykLaarhoven:2015,LvJosephsonWang:2007}, product quantization/clustering methods \cite{JegouDouzeSchmid:2011,JohnsonDouzeJegou:2021,DouzeGuzhvaDeng:2026}, and tree-based methods \cite{BawaCondieGanesan:2005,BeygelzimerKakadeLangford:2006,AndoniRazenshteynNosatzki:2017}.
    More recently, \emph{graph-based methods} have emerged as a popular option for high-dimensional search, performing well on  ANN benchmarks and competitions \cite{AumullerBernhardssonFaithfull:2020,SimhadriWilliamsAumuller:2022,SimhadriAumullerDouze:2026,ManoharShenBlelloch:2024}. Such methods build on ideas dating back to the 1990s \citep{AryaMount:1993,Clarkson:1994,Navarro:1999,KrauthgamerLee:2004}, and include the Hierarchical Navigable Small World Graph method (HNSW) \cite{MalkovYashunin:2020}, Microsoft's DiskANN \cite{KrishnaswamyManoharSimhadri:2024,SubramanyaDevvritKadekodi:2019}, the Navigating Spreading-out Graph method (NSG)  \cite{FuXiangWang:2019}, and others \cite{FuCai:2016,MalkovPonomarenkoLogvinov:2014,HarwoodDrummond:2016}. 
    
    While graph-based methods vary in details, they follow a common theme: a directed graph $G$ is constructed with one node corresponding to every point in $P$. Queries are serviced by running greedy search on the graph. In particular, we start at some $p\in P$, then move to $p$'s out-neighbor in $G$ that gets us closest to $q$, repeating this process until no further improvement is possible.  Typically, a ``back-tracking'' variant of greedy search called \emph{beam search} is used in practice to avoid getting stuck in local minima.

    % These methods build on ideas dating back to the 1990s \citep{AryaMount:1993,Clarkson:1994,Navarro:1999,Kleinberg:2000,krauthgamer2004navigating}

    \subsection{Graph Navigability}
    Given the simplicity of the greedy search process, a key differentiator between graph-based methods is how the search graph $G$ is constructed.\footnote{Some methods, like HNSW, actually utilize a hierarchy of multiple search graphs, running greedy search at each level of the hierarchy.} There is a trade-off: higher degree graphs typically yield more accurate results, but also take more space to store. Moreover, since the per-iteration cost of greedy search scales with the degree of the current node, high degree leads to slower search time.

    While most existing constructions are heuristic, there has been recent interest in formally defining desirable search graph properties, and then separately designing algorithms to construct \emph{sparse graphs} with those properties \cite{ConwayDhulipalaFarach-Colton:2026,KhannaPadakiWaingarten:2025,IndykXu:2023}. Of particular interest is the property of \emph{navigability}, which dates back to work on Stanley Milgram's ``small-world'' phenomenon \cite{Milgram:1967,TraversMilgram:1977,BogunaKrioukovClaffy:2009,ClausetMoore:2003,Kleinberg:2000}. Informally, navigability means that, if a simple greedy search is initialized at any point $p\in P$, it will successfully find any ``in-distribution'' query $q\in P$ (if $q$ is in $P$, it is of course its own nearest neighbor).
    
    Under the standard assumption that distances between points in $P$ are unique\footnote{Without this assumption, ``greedy search'' is not fully specified unless we define a tie-breaking rule. It is simpler to assume unique distances, as this can be ensured, e.g., by adding an arbitrarily small random perturbation to the points in $P$.}, navigability has the following equivalent definition:

     \begin{definition}[Navigable Graph]\label{def:navigability}
            Let $P = \{p_1, \dots, p_n\}$ be a set of points and let $d: P \times P \to \R^{\geq0}$ be any distance function satisfying $d(p,p) = 0$ for all $p \in P$ and $d(p,r) > 0$ for all $p,r\in P$, $p\neq r$. 
            
            A graph $G = (P, E)$ is \emph{navigable} if, for all pairs $p,r \in P$, $p\neq r$, there is a directed edge $(p, s) \in E$ such that $d(s, r) < d(p, r)$.
    \end{definition}
In words, navigability demands that every node $p$ has some out-edge that brings it closer to any other node $r$ in the dataset. That out-edge could be to $r$ itself, so we observe that the complete graph is trivially navigable. 
   
   Navigability is an attractive property for a number of reasons. First, a desirable property of any ANN method is that, for any query $q$, the method returns some approximate nearest neighbor $\tilde{x}$ satisfying $d(\tilde{x},q) < C\cdot \min_{x\in P} d(x,q)$ for some approximation factor $C > 1$. This is the guarantee promised, e.g., by locality sensitive hashing methods. Navigability ensures that this guarantee at least holds for all $q\in P$: in this case, the right hand side is $0$, so greedy search must return $q$ itself to obtain a multiplicative approximation. 

   Moreover, while navigability alone does not ensure good accuracy when $q\notin P$, natural strengthenings of the definition do, including $\alpha$-reachability \cite{IndykXu:2023,SubramanyaDevvritKadekodi:2019,GollapudiKrishnaswamyShiragur:2025} and $\tau$-monotonicity \cite{PengChoiChan:2023}. Multiplicative approximation can also be ensured by combining a navigable graph with the Adaptive Beam Search algorithm from \cite{Al-JazzaziDiwanGou:2025}.

   Finally, beyond theoretical motivation, navigability is relevant practically. Many existing methods target the construction of navigable graphs \cite{SubramanyaDevvritKadekodi:2019}, and indeed the property lends its name to popular methods like the Hierarchical Navigable Small World Graph method (HNSW) and the Navigating Spreading-out Graph method (NSG).

   \paragraph{Limitations of Navigability}
   Despite its current importance in the literature on graph-based search, navigability has a number of limitations. First, it was recently established that, for any distance function, all datasets admit a navigable graph with $O(n\sqrt{n})$ edges \cite{ConwayDhulipalaFarach-Colton:2026}. While far sparser than the complete graph, this bound is also known to be tight, even for random points in $O(\log n)$ dimensions \cite{DiwanGouMusco:2024}. Sparse graphs can be obtained under further data assumptions like bounded doubling dimension \cite{IndykXu:2023,Har-PeledRaichelRobson:2026}, but we find that, for modern billion node datasets, constructing a navigable graph requires hundreds of edges per node, which is much higher than the degree of search graphs used in practice (see \Cref{sec:experiments} for details).

    Beyond high degree requirements, which drive up index size and search time, constructing a navigable graphs is expensive. While recent progress shows how to build near-optimally sparse navigable graphs in $\tilde{O}(n^2)$ time, it was also established that no faster runtime is possible under the Strong Exponential Time Hypothesis. Even for points in $O(\log n)$ dimensional Euclidean space, for any constant $\epsilon$, constructing a navigable graph with $o(n^{2-\epsilon})$ edges (only slightly better than the complete graph) requires $\Omega(n^{2-\epsilon})$ time \cite{ConwayDhulipalaFarach-Colton:2026,KhannaPadakiWaingarten:2025}.

    \subsection{Our Contributions}
    Given the above limitations, it seems natural to consider \emph{relaxed} variations of the navigability property that allow for sparser graphs, faster construction times, or both. Indeed, few practical graph-based methods construct \emph{truly navigable graphs}. Instead, it is common to use fast heuristics that ideally return something ``close'' to navigable. An important example is Microsoft's DiskANN method. If initialized with a complete graph the ``Robust Prune'' algorithm introduced in that work would indeed return a navigable graph. However, doing so is too slow, so a heuristic initialization is used instead \cite{SubramanyaDevvritKadekodi:2019}.

    Nevertheless, to the best of our knowledge, there have not been efforts to quantify how well existing methods approximate the navigability property, or to design methods targeting a specific notation of approximation. In this work, we address that gap by introducing and studying one possible relaxation of navigability. In particular, we define the notation of $\gamma$-almost navigability:
    \begin{definition}[$\gamma$-Almost Navigable Graph]\label{def:almost-navigability}
        Let $P = \{p_1, \dots, p_n\}$ be a set of points and let $d: P \times P \to \R^{\geq0}$ be any distance function satisfying $d(p,p) = 0$ for all $p \in P$ and $d(p,r) > 0$ for all $p,r\in P$, $p\neq r$. For a direct graph $G = (P, E)$, let 
            $$
                C_p = \{r \in P \mid  \exists (p,s) \in E \text{ such that } d(s, r) < d(p, r)\}.
            $$
        
        For a parameter $\gamma \in [0,1]$, $G$ is \emph{$\gamma$-almost navigable} if for every point $p \in P$, $|C_p| \geq \gamma \cdot (n-1)$.
    \end{definition}
    In words, $\gamma$-almost navigability demands that every node has an out edge that takes it closer to a $\gamma$ fraction of the other nodes in the dataset. Setting $\gamma = 1$ recovers navigability. The definition becomes easier to satisfy for smaller $\gamma$. 

    Our first result is that this natural relaxation significantly reduces the worst-case density required to construct a search graph. In contrast to fully navigable graphs, which require $\Omega(n^{3/2})$ edges \cite{DiwanGouMusco:2024}, we show that any data set, under any distance function, admits a $\gamma$-almost navigable graph with just $O(n)$ edges. The leading constant scales with $1/(1-\gamma)$.

    \begin{restatable}{theorem}{almostNavigableExistance}\label{lem:almost-navigable-existence}
        For any point set $P$, distance function $d: P \times P \to \R^{\geq 0}$, and $\gamma \in [0,1)$, there exists a $\gamma$-almost navigable graph with average out-degree at most $\frac{4}{1-\gamma}$. 
    \end{restatable}

    \Cref{lem:almost-navigable-existence} is proven in \Cref{sec:existance_proof} via a constructive argument. We show how to build a linear-sized almost navigable graph using a ``clique peeling'' method that is reminiscent of recent techniques used to construct fully navigable graphs \cite{ConwayDhulipalaFarach-Colton:2026}. We start by partitioning $P$ into arbitrary sets of $O(1/(1-\gamma))$ points. For each such set, $S$, we add a clique to the search graph $G$. After doing so, any $p\in S$ has an edge closer to all points in the dataset \emph{except those for which it is the closest point in $S$.} On average, each point $p\in S$ is closest to  $(1-\gamma)n$ points, so an expectation argument shows that at least half of the points in $S$ have edges closer to a $\gamma$ fraction of points in $P$. I.e., half of the points satisfy the constraint of \Cref{def:almost-navigability}.
    
    We can then repeat the construction, arbitrarily grouping any remaining points that do not satisfy the constraints of \Cref{def:almost-navigability} and adding a new set of cliques. After $O(\log n)$ rounds, all points will have edges close to a $\gamma$ fraction of points in $P$, as required. 

    Naively, the above procedure runs in ${O}(n^2)$ time, as we need to compute distances between all pairs of points to determine which points have satisfied the $\gamma$-almost navigability requirement, and which points should continue to the next round of the algorithm. However, if we introduce randomization, this cost can be reduced: checking a sample of roughly $O(1/(1-\gamma))$ points suffices to determine, with high probability, if a given points $p$ has an edge closer to a $\gamma$-fraction of points in the dataset. This leads to an algorithm whose runtime scales \emph{linearly} instead of quadratically with $n$:
       \begin{restatable}{theorem}{almostNavigableConstruction}\label{thm:sparse-almost-navigable-construction}
        There is an algorithm that, given any point set $P$, distance function $d: P \times P \to \R^{\geq 0}$ computable in time $T$, $\gamma \in [0,1)$, and $\delta\in (0,1)$, constructs a $\gamma$-almost navigable graph with average out-degree $O\left(\frac{1}{1-\gamma}\right)$ in ${O}\left(\frac{nT\log(n/\delta)}{1-\gamma}\right)$ time, with prob. at least $1-\delta$.
    \end{restatable}
    \Cref{thm:sparse-almost-navigable-construction} is proven in \Cref{sec:linear_time}. Combined with \Cref{lem:almost-navigable-existence}, it establishes polynomial improvements in graph size and construction time over fully navigable graphs. In particular, $O(n^{3/2})$ size and $\tilde{O}(n^2)$ construction time are reduced to $O(n)$ and $\tilde{O}(n)$, respectively, when we relax the definition to $\gamma$-almost navigability.

    \paragraph{Value for Greedy Search} Given these efficiency gains, it remains to determine if $\gamma$-almost navigable graphs preserve the  performance of fully navigable graphs for approximate nearest neighbor search. We explore this question both theoretical and empirically. 

    On the theoretical side, we show a negative result in \Cref{sec:hard_instance}. Recall that a desirable property of navigability is that greedy search run on a navigable graph is at least guaranteed to correctly answer any ``in-distribution'' query, $q\in P$. This is a necessary property for good ANN performance, even if it is not sufficient in the worst-case. 
    
    It is natural to ask if $\gamma$-almost navigability leads to a relaxed version of this guarantee. For example, we might hope that greedy search correctly answers a large fraction of in-distribution queries.
    Unfortunately, we show that this is not the case, even for $\gamma$ very close to $1$, and even for points in low-dimensional Euclidean space: we construct a family of 2D point sets and corresponding $\gamma$-almost navigable graphs so that greedy search fails on all but $1/(1-\gamma)$ in-distribution queries, no matter how large $n$ is. 

    Our work on the empirical side, however, is more positive. We show that, on standard ANN benchmarking datasets, $\gamma$-almost navigable graphs achieve similar search quality at much lower search and graph storage costs in comparison to fully navigable graphs. For example, for a fixed target accuracy of $90\%$ recall, we find that $\gamma$-almost navigable graphs require about 50\% less storage space and incur just 35-47\% of the search costs of navigable graphs. \Cref{sec:experiments} contains more details about the practical performance of almost navigable graphs and our experimental setup.

    % \textbf{\textcolor{blue}{Need to fill in.}}
    
    \section{Almost Navigable Graphs}
    In this section, we restate and prove our main theoretical results on the existence and construction of sparse \emph{almost} navigable graphs (\Cref{def:almost-navigability}). We also provide a construction showing that, unfortunately, almost navigable graphs do not enjoy the same worst-case greedy search guarantees as fully navigable graphs. Nevertheless, as shown in the next section, these graphs still appear empirically useful for greedy graph-based ANN search. 

        \subsection{Bounded Degree Almost Navigable Graphs}
        \label{sec:existance_proof}
        Our first result is that, for fixed $\gamma$, it is always possible to construct a $\gamma$-almost navigable graph whose edge count grows just linearly in the dataset size. This contrasts with the necessary $O(n^{3/2})$ edges required to construct a fully navigable graph in the worst-case \cite{DiwanGouMusco:2024}:

      \almostNavigableExistance*

     Our proof is constructive: we show how to build a $\gamma$-almost navigable graph with average out-degree $O(1/({1-\gamma}))$ using a technique similar to the navigable graph construction from \cite{ConwayDhulipalaFarach-Colton:2026}. We  make the construction efficient (near linear time) in \Cref{sec:linear_time}.

      Specifically, our construction leverages what \cite{ConwayDhulipalaFarach-Colton:2026} calls the ``power of cliques'' idea. We begin by partitioning our dataset $P$ into \emph{arbitrary} sets of size $\lceil 2/(1-\gamma) \rceil$. For each such set $S$, we add a clique to $G$, i.e. we add a directed edge from $u$ to $v$ for all $u,v\in S$. While doing so only adds $O(n/(1-\gamma))$ edges to the graph, we claim that after adding these cliques, the requirements of $\gamma$-almost navigability are satisfied for at least half of the points in $P$. Concretely:
        \begin{claim}\label{claim:clique-ownership}
            Let $P$ be set of $n$ points and suppose $G = (P,E)$ contains a clique connecting all nodes in some subset $S \subseteq P$. For all $v\in S$, define:
            \begin{align*}
             U_{S,v} := \{p \in P \mid v = \argmin_{x \in S} d(p, x)\}
            \end{align*}
            Then, for at least half the points in $S$,  have $|U_{S,v}| \leq {2n}/{|S|}$.
        \end{claim}
        Observe that, for all $v\in S$, for all $r \notin U_{S,v}$, there is some edge $(v,s) \in E$ such that $d(s,r) < d(v,r)$: simply take $s$ to be the point in $S$ that is closest to $r$. This means that the number of nodes in $P\setminus{v}$ that $v$ does not have an edge closer to is at most $|U_{S,v}|-1$.\footnote{$v$ itself is in $U_{S,v}$, so there are only $|U_{S,v}|-1$ points from $P\setminus v$ in $U_{S,v}$.} If we set $|S| = O(1/(1-\gamma))$, \Cref{claim:clique-ownership} ensures that $|U_{S,v}| \leq (1-\gamma)n$. We can conclude that at least half of the points in $S$ have edges closer to a $\gamma$ fraction of the remaining points in $P$, as required by \Cref{def:almost-navigability}.
        \begin{proof}[Proof of \Cref{claim:clique-ownership}]
            Observe that the ``ownership'' sets $U_{S,v}$ partition the dataset $P$. We thus have:
            \begin{align*}
                \frac{1}{|S|} \sum_{v \in S}|U_{S,v}| = \frac{|P|}{|S|} = \frac{n}{|S|}.
            \end{align*}
            Using the simple fact that for positive numbers, the median is at most twice the mean, the median sized set has size at most $2n/|S|$. Therefore, for at least half of $v\in S$ have $|U_{S,v}| \leq 2n/|S|$.
        \end{proof}
        \Cref{claim:clique-ownership} provides a simple way of ensuring that the majority of points in $P$ satisfy the requirements of $\gamma$-almost navigability. We obtain our main result by applying this idea iteratively to construct a full $\gamma$-almost navigable graph. After the first round of adding cliques to $G$, we set aside all points that satisfy the $\gamma$-almost navigability requirements. We then group any remaining points that do not satisfy the constraint into a new set of cliques. After doing so, we will again have that half of the remaining points are satisfied. Continuing in this way, we will have a $\gamma$-almost navigable graph after $O(\log n)$ rounds. We formally analyze this procedure below:
        \begin{proof}[Proof of \Cref{lem:almost-navigable-existence}]
            Specifically, consider the following procedure for constructing a graph $G = (P,E)$:
            \begin{enumerate}
                \item Arbitrarily partition $P$ into subsets $S_1, \dots, S_k$ of size $\lceil 2/(1-\gamma)\rceil$ each  and a set $\bar{S}$ with $< \lceil 2/(1-\gamma)\rceil$ left over points.
                \item Add a clique to each $S_i$.
                \item Let $\bar P$ contain all points in $\bar{S}$ and, for every $S_i$, all $v \in S_i$ such that $|U_{S_i,v}| > (1-\gamma) n$. 
                \item If $|\bar P| \geq \lceil 2/(1-\gamma)\rceil$, remove all out-edges from every $v \in \bar P$ and repeat from Step 1 with points in $\bar P$.
                \item Otherwise, add an edge from the $< \lceil 2/(1-\gamma)\rceil$ points in $\bar{P}$ to every point in $P$.
            \end{enumerate}
            We first establish correctness of the above procedure. First, observe that, by \Cref{claim:clique-ownership}, the procedure terminates, and indeed terminates after at most $O(\log n)$ rounds. In particular, by \Cref{claim:clique-ownership}, for each $S_i$, at least $|S_i|/2$ points $v\in S_i$ satisfy:
            \begin{align*}
                |U_{S_i,v}| \leq \frac{2n}{|S_i|} \leq (1-\gamma)n.
            \end{align*}
            Even counting the left over points in $\bar{S}$, it follows that the size of $\bar{P}$ shrinks by at least a factor of $1/4$ at every iteration of the procedure. 
            
            Having established convergence, consider any point $v\in P$. If $v$ remains in $\bar{P}$ until the last step of the algorithm, then $v$ clearly has edges satisfying the $\gamma$-almost navigability requirements of \Cref{def:almost-navigability}, as it is connected to all of $P$.

            If not, then at some point in the algorithm, $v$ was placed in a subset $S_i$ and it was determined that $|U_{S_i,v}| \leq (1-\gamma)n$. When this happens, an edge is added from $v$ to every point in $S_i$ and is never removed. $v$ thus has an edge closer to any $r\in P$ that is not in $U_{S_i,v}$. We conclude that $v$ \emph{does not} have an out edge closer to at most $|U_{S_i,v} \setminus \{v\}| \leq (1-\gamma)n - 1 \leq (1-\gamma)(n-1)$ points in $P$. The $\gamma$-almost navigability requirement of \Cref{def:almost-navigability} thus holds for $v$.

            Finally, we consider the total number of edges in $G$. We have at most $\left\lfloor \frac{2}{1-\gamma}\right\rfloor$ points in $\bar{P}$ at the end of the procedure, each with $n-1$ edges. Any other point $v$ that is not in $\bar{P}$ at the end of the procedure has exactly $\lceil 2/(1-\gamma)\rceil$ edges. So in total we have:
            \begin{align*}
            |E| < \left\lfloor \frac{2}{1-\gamma}\right\rfloor \cdot (n-1) + \left\lceil \frac{2}{1-\gamma}\right\rceil  \cdot n \leq \frac{4}{1-\gamma} \cdot n.\quad \quad\qedhere
            \end{align*}
        \end{proof}

        It is interesting to ask if \Cref{lem:almost-navigable-existence} can be improved. Notably, in the extreme case when $\gamma > 1 - 1/\sqrt{n}$, we have $O(n/(1-\gamma)) > O(\sqrt{n})$, so we can actually obtain a better bound on the degree required for $\gamma$-almost navigability by appealing to the existing upper bounds of $O(\sqrt{n})$ for the average degree of a \emph{fully} navigable graph. This suggests that it might be possible to obtain an improved dependence on $\gamma$. We note that a bound of the form $O({n}/\sqrt{1-\gamma})$ is not possible, however, as it can be checked that the lower bound of \cite{DiwanGouMusco:2024} for navigable graph construction extends to the setting when only a $1 - O(1/\sqrt{n})$ fraction of constraints need to be satisfied. I.e., a $(1-O(1/\sqrt{n}))$-almost navigable graph requires $\tilde{\Omega}(\sqrt{n})$ degree in the worst case, even for points in Euclidean space.

        \subsection{Almost Navigable Graphs in Linear Time}
        \label{sec:linear_time}
        We next establish that, in contrast to fully navigable graphs, sparse almost-navigable graphs can be constructed in sub-quadratic time (actually, in linear time). We restate the formal result below:
        
        \almostNavigableConstruction*

        This result is proven via an efficient implementation of the procedure described in the proof of \Cref{lem:almost-navigable-existence}. The most expensive operation in each round of that procedure is computing the ``ownership'' set of each point $v \in S_i$, $U_{S_i,v}$. Exactly computing this set requires pairwise distance computations between every $v$ with every other point in $P$, resulting in $O(n^2)$ distance computations.
        
        However, we do not actually need $U_{S_i,v}$ itself: we only need to know the \emph{size} of the set to determine if $v$ proceeds to the next round of the procedure. The size $\left| U_{S_i,v}\right|$ can be approximated more efficiently based on a random sample of data points from $P$. In particular, we argue that $\tilde{O}(1/(1-\gamma))$ samples from $P$ suffice by a standard Chernoff bound. The distance computations from each $v\in P$ to these samples dominates the runtime. 
        Our approach is formalized in \Cref{alg:sparse-almost-navigable} and analyzed in detail below:
            
            \begin{algorithm}[h]
                \caption{Almost navigable graph construction}\label{alg:sparse-almost-navigable}
                \begin{algorithmic}[1]
                    \vspace{-.1em}
                    \Require Set of $n$ points $P$, blackbox access to distance function $d: P \times P \to \R^{\geq0}$, $\gamma \in [0,1)$, failure probability $\delta \in (0,1)$.
                    \Ensure $\gamma$-navigable graph $G = (P,E)$ with probability $1-\delta$.
                    \algrule
                    \State Initialize $E \gets \emptyset$, $\Pi^{(0)} \gets P$, $i \gets 0$, $w \gets \frac{16\log(n/\delta)}{1-\gamma}$.
                    \While{$|\Pi^{(i)}| \geq \lceil 4/(1-\gamma) \rceil$}
                        \State Arbitrarily partition $\Pi^{(i)}$ into subsets $S_1, \dots, S_k$ of size $\left\lceil \frac{4}{1-\gamma} \right\rceil$ and a set $\bar{S}$ with $< \left\lceil \frac{4}{1-\gamma} \right\rceil$ left over points.
                        \State Initialize $\Pi^{(i+1)} \gets \bar{S}$
                        \State Draw a multiset, $W$, consisting of $w$ points from $P$ selected uniformly at random with replacement. 
                        \For {$j \in 1, \dots, k$}
                        \For {$v\in S_j$}
                            \State $\hat U_{S_j, v} \gets \{p \in W \mid v = \argmin_{x \in S_j} d(p,x)\}$
                            \If{$|\hat U_{S_j, v}| \leq (1-\gamma) w/2$}
                                \State Add $(v,u)$ to $E$ for every $u\in S_j$. 
                            \Else
                                \State Add $v$ to $\Pi^{(i+1)}$.
                            \EndIf
                            % \State Sort $v \in S_j$ by size of $U_v$ and add top half of points to $P^{(i+1)}$
                        \EndFor
                        \EndFor
                        \State $i \gets i + 1$
                    \EndWhile
                    
                    \State Connect any remaining points in $\Pi^{(i)}$ to all points in $P$.
                    \State \Return $G = (P, E)$
                \end{algorithmic}
            \end{algorithm}

            \begin{proof}[Proof of \Cref{thm:sparse-almost-navigable-construction}]
                We first analyze the correctness of \Cref{alg:sparse-almost-navigable} assuming that the algorithm terminates after $c\log n$ iterations of the main while loop for a fixed constant $c$. We then bound the number of iterations and the running time of the algorithm. 

                \paragraph{Correctness.} To establish correctness, we need to prove that every point $v$ that is \emph{not} added to $\Pi^{(i+1)}$ during round $i$ of the algorithm satisfies the constraints of $\gamma$-almost navigability. As in the proof \Cref{lem:almost-navigable-existence}, since any such $v$ is connected to all other $u \in S_j$ at the end of the round, it suffices to show that $|U_{S_j,v}| \leq (1-\gamma)n$, where 
                \begin{align*}
                U_{S_j,v} := \{p \in P \mid v = \argmin_{x \in S} d(p, x)\}
                \end{align*}
                Since $v$ in not added to $\Pi^{(i+1)}$ if $|\hat U_{S_j,v}| \leq (1-\gamma) w/2$, our task reduces to showing that, if $|U_{S_j,v}| > (1-\gamma)n$, then $|\hat U_{S_j, v}| > (1-\gamma) w/2$ with high probability. This follows from standard concentration bounds. 
                It particular, observe that:
                \begin{align*}
                |\hat U_{S_j, v}| = \sum_{p\in W} \mathbbm{1}[p \in U_{S_j,v}] = \sum_{i=1}^w x_i, 
                \end{align*}
                where each $x_i$ is a Bernoulli random variable with mean $|U_{S_j,v}|/n$. By a Chernoff bound, we thus have that:
                \begin{align*}
                \Pr\left[|\hat U_{S_j, v}| \leq \frac{w}{n}|U_{S_j,v}|/2\right] \leq e^{-w|U_{S_j,v}|/8n}.
                \end{align*}
                Setting $w = \frac{16\log(n/\delta)}{1-\gamma}$, we conclude that, if $|U_{S_j,v}| > (1-\gamma)n$, $|\hat U_{S_j, v}| > (1-\gamma) w/2$ with probability at least $1-\delta/n^2$. Taking a union bound over all $v$, over all $c\log n$ rounds of the algorithm, we conclude that, with probability at least $1-\delta$, every point $v$ with $|U_{S_j,v}| > (1-\gamma)n$ at any round $i$ was correctly added to $\Pi^{(i+1)}$. As discussed above, the graph $G$ is thus $\gamma$-almost navigable. 

                \paragraph{Runtime.} It remains to bound the 
                number of iterations of the main while loop, and consequently the runtime of the algorithm. To do so, observe that, for any $S_j$, at any iteration, $\sum_{v\in S_j} |\hat U_{S_j, v}| = n$. As in the proof of \Cref{claim:clique-ownership}, since $S_j$ has size $\left\lceil {4}/{(1-\gamma)}\right\rceil$, it follows that at least half of $v\in S_j$ have $|\hat U_{S_j, v}| \leq (1-\gamma) w/2$. Accordingly, at iteration of the while loop, at most half of the points in $S_j$ are added to $\Pi^{(i+1)}$. Conservatively accounting for $\bar{S}$, we conclude that:
                \begin{align*}
                |\Pi^{(i+1)}| &\leq \frac{3}{4}|\Pi^{i}| & &\text{ for all } i.
                \end{align*}
             Accordingly, the while loop terminates after at most $3 \log_2 n$ iterations. Moreover, the cost of \Cref{alg:sparse-almost-navigable} is dominated by the $|\Pi^{i}| \cdot w$ distance computations performed at each round $i$ of the while loop. Since the size of $|\Pi^{(i+1)}|$ is decreasing geometrically, we just conclude an overall runtime of:
             \begin{align*}
                O(n w T) = O\left(\frac{n\log(n/\delta)T}{1-\gamma} \right), 
             \end{align*}
             where $T$ is the cost of a single distance computation.
            \end{proof}
            
        \subsection{Performance Under Greedy Search}
        \begin{figure*}[tb]
                \centering
                \includegraphics[width=0.8\linewidth]{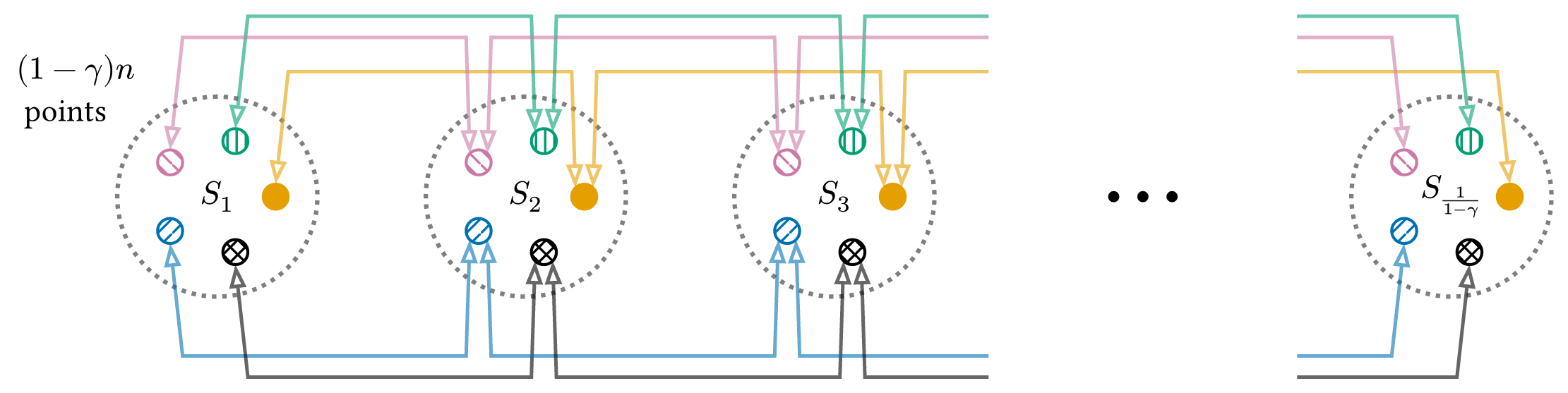}
                \caption{This figure illustrates the hard instance used to prove \Cref{lem:hard_instance}. Each cluster, $S_1, \ldots, S_k$, consists of $(1-\gamma)n$ points and the edges comprise a $\gamma$-almost navigable graph for the dataset. However, if we initialize greedy search at any point, $s$, it can only reach points with the same color as $s$. We will thus fail on all but $k = \frac{1}{1-\gamma}$ in-distribution queries.}
                \label{fig:hard-instance}
            \end{figure*}
        \label{sec:hard_instance}
        With \Cref{lem:almost-navigable-existence,thm:sparse-almost-navigable-construction} in place, we have established that almost navigable graphs require much lower degree and less construction time than their fully navigable counterparts. It remains to determine if these relaxed graphs remain valuable for ANN search. 

        In the next section, we address this question empirically, establishing overall positive results. However, we first provide some initial evidence that theoretically understanding the performance of almost navigable graphs maybe be even more difficult than understanding fully navigable graphs. In particular, the limited theoretical guarantees available for fully navigable graphs do not extend to almost navigable graphs, even approximately. 

        To be more concrete, recall that full navigability at least ensures correctness for \emph{in-distribution queries}. If we receive a query $q$ that happens to be in the dataset $P$, then  greedy search on a navigable graph $G$ will return $q$ itself. While this property alone does not imply good performance on out-of-distribution queries $q \notin P$, it is a natural starting point for understanding search accuracy.

        We might hope to show that $\gamma$-almost navigable graphs enjoy a similar guarantee. Perhaps we do not succeed \emph{for all} in-distribution queries, but for some large fraction that depends on $\gamma$. Unfortunately, we prove that this is not the case:
        \begin{lemma}
            \label{lem:hard_instance}
                For any $n$ and $\gamma \in [0,1)$, there is a set of $n$ points, $P$, in 2-dimensional Euclidean space with a $\gamma$-almost navigable graph $G$ such that greedy search on $G$, originating at any point $s \in P$, will not correctly return $q$ for $n - \frac{1}{1- \gamma}$ in-distribution queries $q \in P$.
        \end{lemma}
            \begin{proof}
                We prove the result using a hard instance that is illustrated in \Cref{fig:hard-instance}. 
                To construct $P$, we place balls $S_1, \dots, S_k$ of radius $\epsilon$ on a straight line at intervals of length greater than $2\epsilon$ and then place $(1 - \gamma) \cdot n$ points in each ball $S_i$. Here, $k = \frac{1}{1- \gamma}$. Next, pick $(1 - \gamma) \cdot n$ colors $\mathcal{X} = \{x_1, \dots, x_{(1 - \gamma) n}\}$ and arbitrarily assign colors so that no two points in the same ball share the same color. From here on, we let $p^{x}_{i}$ denote the point in cluster $S_i$ with color $x \in \mathcal{X}$.

                For each $p^{x}_{i}$, add an edge to $p^{x}_{i-1}$ and $p^{x}_{i+1}$. For $i = 1$ or $i = k$, we just add edges to $p^{x}_{2}$ and $p^{x}_{k-1}$, respectively. In other words, $G$ connects all points of like color in adjacent clusters. 

                By triangle inequality, is not hard to see that $G$ is $\gamma$-almost navigable. The two edges out of each $p^{x}_{i}$ allow it to get closer to all points in other clusters. Because there are only $(1-\gamma)n$ points in $S_i$, $p^{x}_{i}$ thus satisfies the $\gamma$-almost navigability requirement.

                On the other hand, notice that  $G$ is disconnected and there are $(1 - \gamma)n$ separate connected components. For any query $q \in P$, greedy search will only return $q$ if its color matches that of the starting node $s$. There are only $n - k = n - \frac{1}{1 - \gamma}$ such points.
            \end{proof}

%                 \textcolor{blue}{\textbf{
% Due June 12th anytime on earth. 8 pages not including citations and appendix.
%     }}
%     Workshop link: \url{https://vecdb-ws.github.io/vldb2026/cfp.html}

    \section{Experiments}
    \label{sec:experiments}        
        We conclude by presenting experimental evaluations of almost navigable graphs on real world datasets, comparing them to their fully navigable counterparts. We demonstrate practical improvements in graph sparsity and in search efficiency when using beam search, the standard back-tracking variant of plain greedy search used in almost all practical graph-based ANN systems. All our experiments use standard Euclidean distance as a distance metric.
        
        \subsection{Graph Density Experiments}
        We begin by comparing the density of navigable and almost navigable graphs on on 8 standard datasets, MNIST, Fashion MNIST, COCO-i2i, Glove25, Microsoft SPACEV1B, Yandex DEEP, BIGANN, and Facebook SimSearchNet \cite{AumullerBernhardssonFaithfull:2020, SimhadriWilliamsAumuller:2022}. Details of these dataset are included in \Cref{sec:appendix-tables-plots}, \Cref{tab:coverage-stats}. 
        
        \paragraph{Experimental Setup.}
        Ideally, we would like to understand the minimum number of edges required to construct a $\gamma$-almost navigable graph for various choices of $\gamma < 1$ and compare to the $\gamma = 1$ (fully navigable) case. While minimizing sparsity is computationally intractable, as established in \cite{ConwayDhulipalaFarach-Colton:2026,KhannaPadakiWaingarten:2025} the problem of constructing a navigable graph amounts to solving a set cover instance for each node in the dataset. Likewise, $\gamma$-almost navigable graph construction involves $n$ \emph{partial} set cover problems. We can thus obtain graphs with sparsity within $\log n$ of optimal by using the standard greedy set cover algorithm to select out edges for each node \cite{Kearns1990}. 

        Even greedy set cover is computationally expensive, however, requiring $O(n^2)$ time for each node. Instead, we implement the popular ``robust prune'' algorithm from Microsoft's DiskANN library \cite{SubramanyaDevvritKadekodi:2019}. Pseudocode is provided in \Cref   {alg:robust-prune}. For a given node $v$, the algorithm first adds an edge to $v$'s nearest neighbor in $P$ and eliminates all nodes in the dataset that it now has an edge closer to. It then adds and edge to its nearest neighbor among the remaining points, continuing until it has eliminated a $\gamma$ fraction of the dataset. Robust prune requires just $O(n\log n)$ time per node and, while it does not have any guarantees, the order of edge additions seems competitive with that of greedy set cover: in initial experiments on small datasets, we only noticed small (1-2 edge) improvements in average degree when running the more expensive greedy method. 

        For the datasets with more than 1 billion points -- Yandex DEEP, BIGANN, Facebook SimSearchNet++, and Microsoft SPACEV1B -- we could not afford to build full graphs, even using robust prune, so we instead construct neighborhoods that satisfy the $\gamma$-almost navigability requirement for 10,000 randomly selected nodes. We report average degree and other statistics over that subset, noting that the average degree is an unbiased estimate of the average degree that would be obtained by the same algorithm run on the full dataset. 

                \begin{figure}[t]
                    \centering
                    \includegraphics[width=0.9\linewidth]{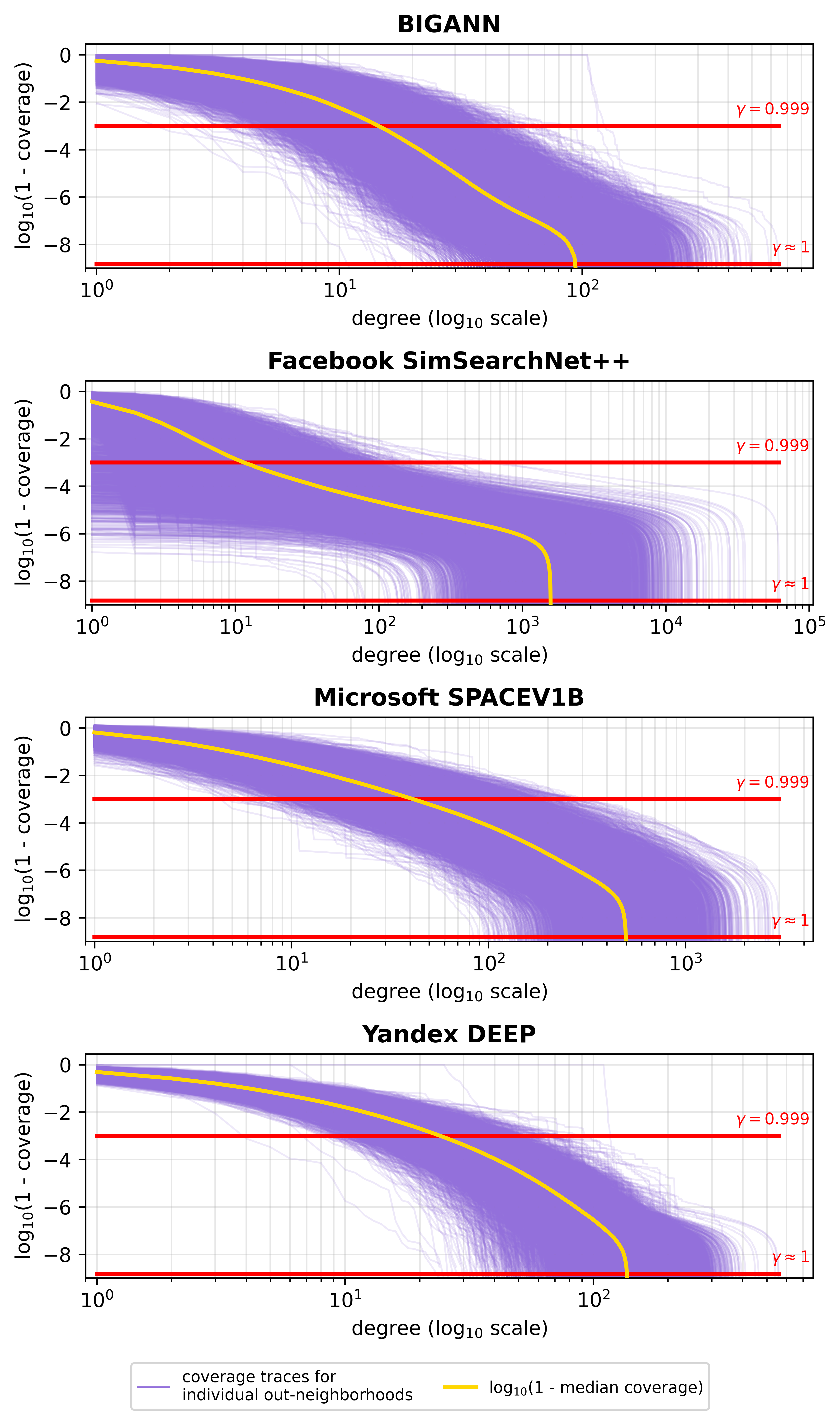}
                    \caption{In these plots, the vertical axis shows the $\log_{10}$ of the fraction of uncovered points plotted against vertex out-degree on the horizontal axis. Each purple line traces graph coverage for a single point as edges are added to its out-neighborhood, and the yellow line describes the median trend. The red horizontal lines mark when coverage $\gamma = 0.999$ and $\gamma \approx 1$ are achieved. The shape of the graphs shows that just the initial few edges account for covering most of the graph, while all subsequent edges cover progressively smaller portions. In fact, the median out-neighborhood uses almost a full order of magnitude fewer edges to achieve coverage = 0.999 compared to the edges needed for full navigability.}
                    \label{fig:degree-vs-coverage}
                \end{figure}

\begin{table}[t]
\centering
\small
\setlength{\tabcolsep}{2pt}
\begin{tabular}{lccccccc}
\toprule
\multirow{2}{*}{\textbf{Dataset}} & \multirow{2}{*}{\textbf{\shortstack{\# of \\ Points}}} & \multirow{2}{*}{\textbf{Dim.}} & \multirow{2}{*}{$\gamma$} & \multicolumn{4}{c}{\textbf{Out-Degree}} \\
\cmidrule(lr){5-8}
& & & & \textbf{Mean} & \textbf{Median} & \textbf{Min} & \textbf{Max} \\
\midrule
\multirow{3}{*}{\shortstack[l]{Fashion\\MNIST \cite{AumullerBernhardssonFaithfull:2020}}} & \multirow{3}{*}{60K} & \multirow{3}{*}{784} & 1 (navigable)\,\,\, & 13.6 & 12 & 1 & 108 \\
\cmidrule(l){4-8}
 &  &  & 0.995 & 6.6 & 6 & 1 & 70 \\
 &  &  & 0.950 & 4.3 & 4 & 1 & 39 \\
\cmidrule(l){1-8}
\multirow{3}{*}{MNIST \cite{AumullerBernhardssonFaithfull:2020}} & \multirow{3}{*}{60K} & \multirow{3}{*}{784} & 1 (navigable)\,\,\, & 19.6 & 19 & 2 & 67 \\
\cmidrule(l){4-8}
 &  &  & 0.995 & 10.3 & 10 & 2 & 32 \\
 &  &  & 0.950 & 5.8 & 5 & 2 & 20 \\
\cmidrule(l){1-8}
\multirow{3}{*}{COCO-i2i \cite{AumullerBernhardssonFaithfull:2020}} & \multirow{3}{*}{113K} & \multirow{3}{*}{512} & 1 (navigable)\,\,\, & 28.6 & 27 & 2 & 131 \\
\cmidrule(l){4-8}
 &  &  & 0.995 & 13.5 & 12 & 2 & 75 \\
 &  &  & 0.950 & 6.7 & 6 & 2 & 52 \\
\cmidrule(l){1-8}
\multirow{3}{*}{Glove25 \cite{AumullerBernhardssonFaithfull:2020}} & \multirow{3}{*}{1.2M} & \multirow{3}{*}{25} & 1 (navigable)\,\,\, & 50.4 & 50 & 1 & 141 \\
\cmidrule(l){4-8}
 &  &  & 0.995 & 8.3 & 7 & 1 & 71 \\
 &  &  & 0.950 & 4.2 & 4 & 1 & 44 \\
\cmidrule(l){1-8}
\multirow{3}{*}{\shortstack[l]{Yandex\\DEEP \cite{SimhadriWilliamsAumuller:2022}}} & \multirow{3}{*}{1B} & \multirow{3}{*}{96} & 1 (navigable)\,\,\, & 144.1 & 138 & 25 & 560 \\
\cmidrule(l){4-8}
 &  &  & 0.995 & 16.3 & 16 & 4 & 115 \\
 &  &  & 0.950 & 6.8 & 6 & 2 & 113 \\
\cmidrule(l){1-8}
\multirow{3}{*}{BIGANN \cite{SimhadriWilliamsAumuller:2022}} & \multirow{3}{*}{1B} & \multirow{3}{*}{128} & 1 (navigable)\,\,\, & 105.5 & 95 & 12 & 647 \\
\cmidrule(l){4-8}
 &  &  & 0.995 & 11.8 & 11 & 2 & 117 \\
 &  &  & 0.950 & 6.3 & 6 & 1 & 112 \\
\cmidrule(l){1-8}
\multirow{3}{*}{\shortstack[l]{Facebook\\SimSearch-\\Net++ \cite{SimhadriWilliamsAumuller:2022}}} & \multirow{3}{*}{1B} & \multirow{3}{*}{256} & 1 (navigable)\,\,\, & 2077.6 & 1577 & 44 & 61399 \\
\cmidrule(l){4-8}
 &  &  & 0.995 & 8.4 & 7 & 1 & 133 \\
 &  &  & 0.950 & 3.9 & 3 & 1 & 33 \\
\cmidrule(l){1-8}
\multirow{3}{*}{\shortstack[l]{Microsoft\\SPACEV1B \cite{AumullerBernhardssonFaithfull:2020}}} & \multirow{3}{*}{1.4B} & \multirow{3}{*}{100} & 1 (navigable)\,\,\, & 554.8 & 500 & 85 & 2984 \\
\cmidrule(l){4-8}
 &  &  & 0.995 & 22.1 & 19 & 3 & 141 \\
 &  &  & 0.950 & 7.4 & 7 & 1 & 76 \\
\bottomrule
\end{tabular}
\smallskip
\caption{This table compares degree statistics for $\gamma$-almost navigable graphs to those for fully navigable graphs ($\gamma = 1$).
}
\label{tab:coverage-stats}
\end{table}

               \paragraph{Results.}
    \Cref{tab:coverage-stats} compares degree statistics between almost navigable graphs ($\gamma < 1$) with those of fully navigable graphs ($\gamma = 1$). Across all the datasets, we observe a marked improvement in graph sparsity, even for very high values of $\gamma$. For $\gamma = 0.995$, which means $99.5\%$ of constraints are satisfied, we saw roughly a 50\% reduction in the mean and median out-degree on the smaller datasets with less than 1 million points, like MNIST, Fashion-MNIST, and COCO-i2i. For larger datasets, there was an even greater reduction, with Glove25, SPACEV1B, Yandex DEEP, BIGANN, and Facebook SimSearchNet++ needing fewer than 20\% edges of the entire navigable neighborhood to achieve 99.5\% coverage. Remarkably, we only needed $\sim 9$ edges on average to achieved 99.5\% navigability for Facebook SimSearchNet++, as opposed to $\sim 2077$ for full navigability. 
    
    For other choices of $\gamma$, in \Cref{fig:degree-vs-coverage}, we compute how graph coverage changes as edges are added to each point's out-neighborhood for each of the billion point datasets. We observe that only the initial few edges picked during graph construction account for satisfying most of the navigability constraints, showing that almost navigability can be achieved using orders of magnitude fewer edges than full navigability. The median number of edges required to achieve $\gamma = 0.999$ almost navigability is 2 orders of magnitude less than the median number of edges required for full navigability for Facebook SimSearchNet++ and 1 order of magnitude for BIGANN, Yandex DEEP, and Microsoft SPACEV1B.

            \begin{figure*}[t]
                    \centering
                    \includegraphics[width=0.95\linewidth]{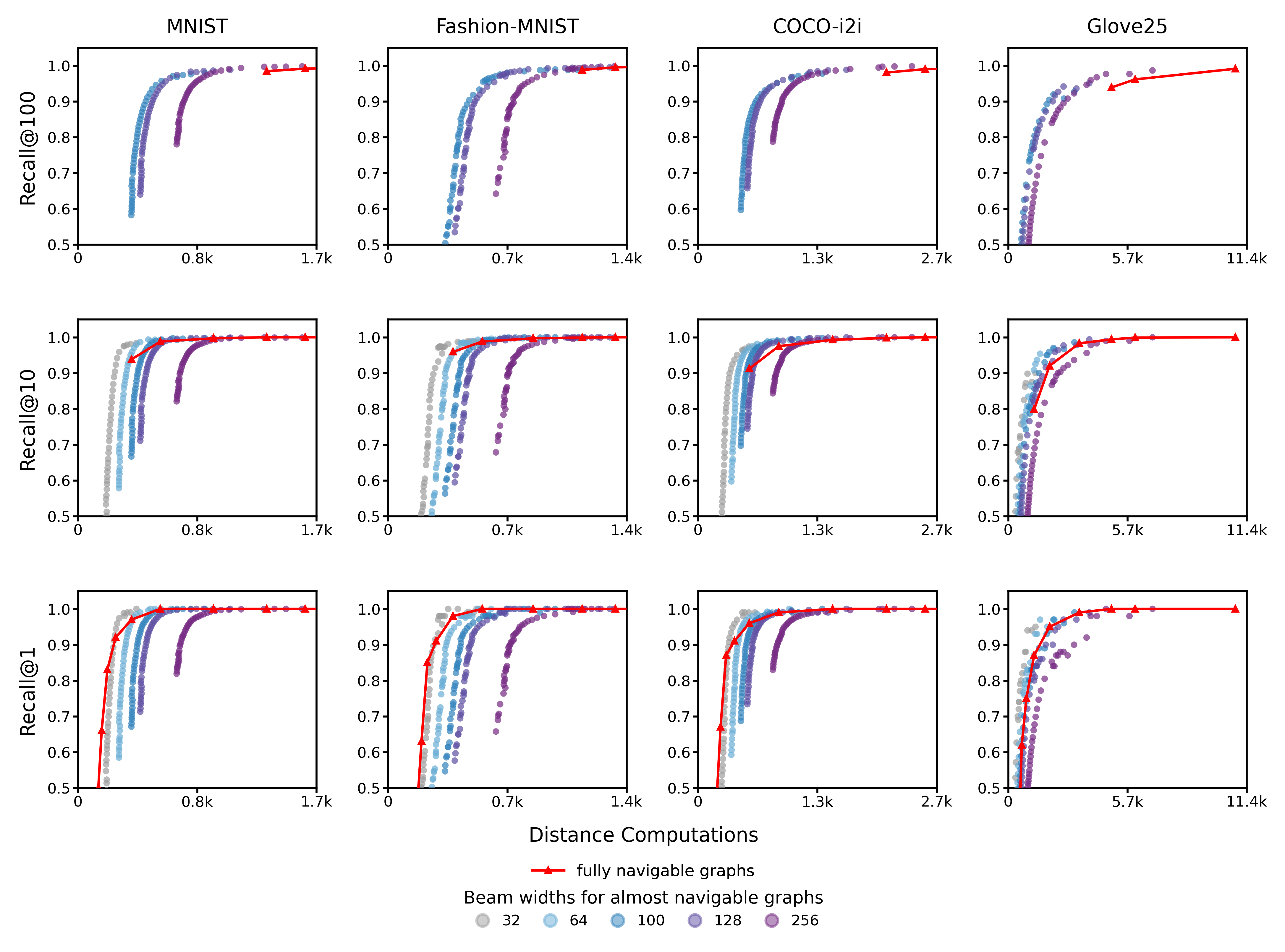}
                    \caption{These plots show average recall$@k$ vs. average distance computations for various choices of navigability parameter, $\gamma$, and beam width parameter, $b$. The red solid line represents the recall curve for fully navigable graphs for various choices of beam width. We only plot points corresponding to beam widths $\geq k$, since having a lower beam width places a limit on the maximum achievable recall. In these plots, points up and to the left correspond to better performance: higher recall and fewer distance computations. As can be seen, for all data sets, there are almost navigable graphs which, for an appropriate choice of beam width, match or beat the performance of search on the fully navigable graph.}
                    \label{fig:recall-v-dc}
                \end{figure*}
            
            \subsection{Retrieval Experiments}
            In addition to graph sparsity, we evaluate nearest neighbor query performance on the smaller of the eight datasets studied above (those for which we were able to construct full search graphs. In particular, we consider the standard recall@k metric: we use the graphs to obtain a set of $k$-nearest neighbors for a given query $q$, then compute the fraction of these answers that are true top-$k$ neighbors of $q$. 
            
            \paragraph{Experimental Setup.}
            To obtain $k$-nearest neighbors, we run \emph{beam search} on the constructed almost navigable and fully navigable graphs Beam search has become the defacto back-tracking variant of greedy search used in most graph-based methods \cite{MalkovYashunin:2020, SubramanyaDevvritKadekodi:2019}. Beam search maintains a list of the closest points found so far for a given query $q$. At each step, it picks the closest point from that list, computing the distance between $q$ and all of that point's out-neighbors. It then updates the list, removing the ``explored'' point. When run with beam width parameter $b$, search terminates when the best node in the list is no closer to $q$ than $b$ points explored previously. Increasing $b$ increases how long the search executes for, so provides a natural way to trade-off between recall and runtime. For a more detailed discussion of the method and full pseudocode, see \cite{Al-JazzaziDiwanGou:2025}.
            
            Since the runtime of graph-based search is dominated by the cost of computing distances between the query $q$ and candidate nearest neighbors, we report the total number of distance computations as our main performance metric. This metric is less noisy and system dependent than, e.g., wall-clock runtime. Concretely, for each graph, for each dataset, we run search using beam widths $b \in \{1, 2, 4, 16, 32, 64, 100, 128, 256\}$. For each beam width we compute the average  recall$@k$, where $k \in \{1, 10, 100\}$, and the average number of distance computations over all dataset test queries.

            \begin{table*}[tb]
                \centering
                \small
                \begin{tabular}{lc>{\centering\arraybackslash}p{1.6cm}>{\centering\arraybackslash}p{1.6cm}>{\centering\arraybackslash}p{1.6cm}>{\centering\arraybackslash}p{1.6cm}>{\centering\arraybackslash}p{1.6cm}>{\centering\arraybackslash}p{1.6cm}}
                \toprule
                \multirow{2}{*}{\textbf{Dataset}} & \multirow{2}{*}{\textbf{Target Recall}$@k$} & \multicolumn{3}{c}{\textbf{Ratio of Distance Computations}} & \multicolumn{3}{c}{\textbf{Ratio of Average Degree}} \\
                \cmidrule(r){3-5} \cmidrule(l){6-8}
                & & \textbf{$k = 1$} & \textbf{$k = 10$} & \textbf{$k = 100$} & \textbf{$k = 1$} & \textbf{$k = 10$} & \textbf{$k = 100$} \\
                \midrule
                \multirow{4}{*}{MNIST} & 0.90 & 0.69 & 0.61 & 0.39 & 0.53 & 0.55 & 0.33 \\
                 & 0.95 & 0.66 & 0.67 & 0.44 & 0.55 & 0.41 & 0.41 \\
                 & 0.97 & 0.62 & 0.62 & 0.51 & 0.58 & 0.53 & 0.53 \\
                 & 0.99 & 0.68 & 0.70 & 0.69 & 0.58 & 0.53 & 0.41 \\
                \midrule
                \multirow{4}{*}{Fashion-MNIST} & 0.90 & 0.69 & 0.62 & 0.44 & 0.55 & 0.50 & 0.38 \\
                 & 0.95 & 0.66 & 0.68 & 0.50 & 0.51 & 0.64 & 0.50 \\
                 & 0.97 & 0.68 & 0.65 & 0.56 & 0.64 & 0.50 & 0.57 \\
                 & 0.99 & 0.63 & 0.79 & 0.70 & 0.72 & 0.77 & 0.64 \\
                \midrule
                \multirow{4}{*}{COCO-i2i} & 0.90 & 0.74 & 0.59 & 0.33 & 0.32 & 0.50 & 0.27 \\
                 & 0.95 & 0.63 & 0.56 & 0.42 & 0.50 & 0.39 & 0.32 \\
                 & 0.97 & 0.56 & 0.65 & 0.52 & 0.58 & 0.32 & 0.53 \\
                 & 0.99 & 0.55 & 0.61 & 0.67 & 0.50 & 0.39 & 0.39 \\
                \midrule
                \multirow{4}{*}{Glove25} & 0.90 & 0.53 & 0.53 & 0.41 & 0.50 & 0.41 & 0.37 \\
                 & 0.95 & 0.66 & 0.56 & 0.70 & 0.60 & 0.41 & 0.60 \\
                 & 0.97 & 0.57 & 0.70 & 0.59 & 0.41 & 0.41 & 0.37 \\
                 & 0.99 & 0.93 & 0.81 & 0.65 & 0.60 & 0.60 & 0.60 \\
                \midrule
                \textbf{Average} &  & \textbf{0.65} & \textbf{0.65} & \textbf{0.53} & \textbf{0.54} & \textbf{0.49} & \textbf{0.45} \\
                \bottomrule
                \end{tabular}
                \smallskip
                \caption{
                    This table compares the ratio of computational cost, measured as average distance computations (DC) and average degree, of achieving a target recall$@k$ for $\gamma$-almost navigable graphs (when $\gamma < 1$) relative to fully navigable graphs ($\gamma = 1$). I.e., for some target recall$@k$, the table shows $\frac{\text{avg. DC for $\gamma < 1$}}{\text{avg. DC for $\gamma = 1$}}$ and $\frac{\text{avg. degree for $\gamma < 1$}}{\text{avg. degree for $\gamma = 1$}}$. Values $< 1$ mean the almost navigable graph had lower cost; values $\geq 1$ mean otherwise. For every target recall, $\gamma$-almost navigable graphs require fewer distance computations ($\sim 35 - 47$\% reduction) at smaller graph sizes ($\sim 46-55$\% reduction) on average. The data points for this table were chosen by interpolating the curves in \Cref{fig:recall-v-dc}.
                }
                \label{tab:target-recall-ratio}
            \end{table*}

            \paragraph{Results.}
                \Cref{fig:recall-v-dc} visualizes the resulting cost/performance tradeoffs. Each point on each plot corresponds to a particular choice of the beam width $b$ (indicated by the point's color) and a value of $\gamma$ between $0.8$ and $1$ that was used to construct the search graph. For each choice of $b$, we can trace a line of points with increasing recall and distance computations, which corresponds to increasing $\gamma$, i.e., increasing the graph coverage. Observe that in each of the plots, there are almost navigable graph search configurations that lie above the recall-distance-computation curve of fully navigable graphs. The improvement is most striking for recall@100, where almost navigable graphs achieve higher or similar recall for much fewer distance computations. 

                For a more detailed analysis, \Cref{tab:target-recall-ratio} describes the efficiency trade-off to achieve certain target recalls. More specifically, for a chosen target recall, we find beam width and $\gamma$ settings for almost navigable graphs that achieve that recall and compare them to their fully navigable counterparts. The table displays the fraction of distance computations performed during search on almost navigable graphs relative to the distance computations performed on navigable graphs to achieve the target recall at recall$@k$. Additionally, \Cref{tab:target-recall-ratio} includes the ratio of the average out-degrees of almost navigable graphs to the average out-degree of navigable graphs. In both of these statistics, a score lower than 1 indicates greater search efficiency for almost navigable graphs. 

                Across all datasets, almost navigable graphs are more efficient at achieving their recall targets in terms of the number of distance computations required to perform search on them and the amount of space required to store them (graph sparsity). On average, $\gamma$-almost navigable graphs used 47\% fewer distance computations and 55\% less space than fully navigable graphs to achieve a given target recall$@100$. We see similar efficiency gains for recall$@1$ and recall$@10$, with an average reduction of 35\% in distance computations and $\sim$ 50\% reduction in graph size. \Cref{tab:target-recall-dist} and \Cref{tab:target-recall-deg} contain a more detailed report of the distance computations and average graph degree statistics we compiled for \Cref{tab:target-recall-ratio}. 
                
                \paragraph{On choosing $\gamma$.} 
                    The dataset's size and structure appear to be important in determining the right choice of $\gamma$. In these preliminary evaluations, $\gamma = 0.995$ is sufficient for COCO-i2i and MNIST to achieve recall$@k$ > 0.97 for all $k$, but not for Glove25 and Fashion-MNIST; those required $\gamma = 0.9995$ or higher to achieve the same performance. Interestingly, as reported in \Cref{tab:target-recall-ratio}, these graphs are still sparse and efficient relative to their fully navigable counterparts. See Appendix \ref{sec:appendix-tables-plots} for more detailed tables containing graph degree statistics.

    \section{Discussion}
    In this work, we introduce a natural relaxation of the navigability property, which has become central in work on graph-based approximate nearest neighbor search. Our ``almost navigable graphs'' have a linear number of edges, and can be constructed in near linear time, both polynomial improvements over full navigability. Moreover, initial experimental results appear promising. Looking towards next steps, we hope to provide a more thorough empirical evaluation for search on almost navigable graphs. On the theoretical side, it would be valuable to explore other relaxations of navigability and related concepts like $\alpha$-shortcut reachability. We are particularly interested in relaxations which, unlike $\gamma$-almost navigability, at least partially preserve some of the theoretical properties of navigable graphs. 

    \section{Acknowledgments}
    Pratyush Avi was partially supported by a GAANN fellowship from the US Department of Education. 
    
    \bibliographystyle{ACM-Reference-Format}
    \bibliography{refs}

\onecolumn
\appendix
    \section{Algorithms}\label{sec:appendix-algs}
        In this section, we present pseudocode for the various algorithms used for the experimental analysis presented in this paper. The robust prune and set cover algorithms described here largely follow their standard descriptions, with the only change being the stopping criteria. Constructing a fully navigable graph would require the algorithm to keep picking edges until there is no point left uncovered. In \Cref{line:robust-prune-condition} of the robust prune algorithm and \Cref{line:set-cover-condition} of the greedy set cover algorithm, we instead stop when the number of uncovered points falls below $(1-\gamma)n$. Notably, this is the same as regular navigability when $\gamma = 1$.
        \begin{algorithm}\label{alg:robust-prune}
            \caption{Robust Prune w/ Early Stopping}
            \begin{algorithmic}[1]
                \vspace{-.1em}
                \Require Set of $n$ points $P$, blackbox access to distance function $d: P \times P \to R^{\geq 0}$, $\gamma \in [0, 1]$
                \Ensure $\gamma$-navigable graph $G = (P, E)$
                \algrule
                \State $E \gets \{\}$
                \For{each $p \in P$}
                    \State $U \gets P \setminus\{p\}$ \label{line:robust-prune-condition}
                    \While{$|U| > (1 - \gamma)n$}
                        \State $v \gets \argmin_{x \in U} d(p,x)$
                        \State Add edge $(p,v)$ to $E$
                        \State $C \gets \{y \in U \mid d(v, y) < d(p, y)\}$
                        \State $U \gets U \setminus \{C\}$, remove covered points from $U$
                    \EndWhile
                \EndFor
                \State \Return $G = (P, E)$
            \end{algorithmic}
        \end{algorithm}

        \begin{algorithm}\label{alg:greedy-set-cover}
            \caption{Partial Greedy Set Cover}
            \begin{algorithmic}[1]
                \vspace{-.1em}
                \Require Set of $n$ points $P$, blackbox access to distance function $d: P \times P \to R^{\geq 0}$, $\gamma \in [0, 1]$
                \Ensure $\gamma$-navigable graph $G = (P, E)$
                \algrule
                \State $E \gets \{\}$
                \For{each $p \in P$}
                    \State $U \gets P \setminus\{p\}$
                    \State For each $x \in U$, define $\mathcal{S}_{p \to x} := \{y \in P \mid d(x, y) < d(p, y)\}$
                    \While{$|U| > (1 - \gamma)n$} \label{line:set-cover-condition}
                        \State $v \gets \argmin_{x \in U} |\mathcal{S}_{p \to x} \cap U|$
                        \State Add edge $(p,v)$ to $E$
                        \State $U \gets U \setminus \{\mathcal{S}_{p \to x}\}$, remove covered points from $U$
                    \EndWhile
                \EndFor
                \State \Return $G = (P, E)$
            \end{algorithmic}
        \end{algorithm}

        % \begin{algorithm}[H]\label{alg:beam-search}
        %     \caption{Beam Search \cite{Al-JazzaziDiwanGou:2025}}
        %     \begin{algorithmic}[1]
        %         \vspace{-.1em}
        %         \Require Graph $G = (P, E)$, start node $s$, distance function $d: P \times P \to R^{\geq 0}$, beam width $b$, query $q$, target number of nearest-neighbors $k$
        %         \Ensure A set of $k$ nodes, $B \subset P$
        %         \algrule
        %         \State $D \gets \{s\}$ \Comment{Set of discovered points}
        %         \State $C \gets \{(s, d(q,s)\}$ \Comment{Min-heap of candidates}
        %         \State $B \gets \{(s, d(q,s)\}$ \Comment{Max-heap of best results}
        %         \While{$C$ is not empty}   
        %             \State $(x, d(q, x) \gets \mathrm{heappop}(C)$
        %             \If{$|B| = k$ and $\mathrm{top}(B) \leq d(q, x)$}
        %                 \State \textbf{break}
        %             \EndIf
        %             \For{all $y \in N_{G}(x)$}
        %                 \If{$y \notin D$}
        %                     \State $D \gets D \cup \{y\}$
        %                     \If{$|B| < k$ or $d(q, y) < \mathrm{top}(B)$}
        %                         \State $\mathrm{heappush}(B, (y, d(q, y)))$
        %                         \State $\mathrm{heappush}(C, (y, d(q, y)))$
        %                         \If{$|B| = k + 1$}
        %                             \State $\mathrm{heappop}(B)$
        %                         \EndIf
        %                     \EndIf
        %                 \EndIf
        %             \EndFor
        %         \EndWhile
        %     \end{algorithmic}
        % \end{algorithm}

    \newpage
    \section{Additional Tables and Plots}\label{sec:appendix-tables-plots}
        The following tables include more details about the characteristics of almost navigable graphs on real world datasets. \Cref{tab:coverage-stats-full} expands on \Cref{tab:coverage-stats} and includes out-degree and in-degree statistics for a larger range of coverage values. \Cref{tab:target-recall-dist} and \Cref{tab:target-recall-deg}, expand on \Cref{tab:target-recall-ratio} and include details about the exact distance computations, average graph degree, and value of $\gamma$ that achieve the desired target recalls. 

\begin{table}[h]
\centering
\setlength{\tabcolsep}{8pt}
\begin{tabular}{lcccccccccc}
\toprule
\multirow{2}{*}{Dataset} & \multirow{2}{*}{Points} & \multirow{2}{*}{Dim.} & \multirow{2}{*}{$\gamma$} & \multicolumn{4}{c}{Out-degree} & \multicolumn{3}{c}{In-degree} \\
\cmidrule(lr){5-8} \cmidrule(lr){9-11}
& & & & Mean & Median & Min & Max & Median & Min & Max \\
\midrule
\multirow{5}{*}{Fashion-MNIST} & \multirow{5}{*}{60,000} & \multirow{5}{*}{784} & 1.0000 & 13.55 & 12.0 & 1 & 108 & 12.0 & 1 & 117 \\
 &  &  & 0.9999 & 12.13 & 11.0 & 1 & 107 & 11.0 & 1 & 70 \\
 &  &  & 0.9995 & 9.72 & 8.0 & 1 & 106 & 9.0 & 1 & 55 \\
 &  &  & 0.9950 & 6.59 & 6.0 & 1 & 70 & 6.0 & 1 & 54 \\
 &  &  & 0.9500 & 4.33 & 4.0 & 1 & 39 & 4.0 & 1 & 51 \\
\midrule
\multirow{3}{*}{MNIST} & \multirow{3}{*}{60,000} & \multirow{3}{*}{784} & 1.0000 & 19.64 & 19.0 & 2 & 67 & 17.0 & 1 & 134 \\
 &  &  & 0.9950 & 10.32 & 10.0 & 2 & 32 & 10.0 & 1 & 63 \\
 &  &  & 0.9500 & 5.79 & 5.0 & 2 & 20 & 6.0 & 1 & 31 \\
\midrule
\multirow{3}{*}{COCO-i2i} & \multirow{3}{*}{113,287} & \multirow{3}{*}{512} & 1.0000 & 28.59 & 27.0 & 2 & 131 & 24.0 & 1 & 268 \\
 &  &  & 0.9950 & 13.54 & 12.0 & 2 & 75 & 13.0 & 1 & 87 \\
 &  &  & 0.9500 & 6.69 & 6.0 & 2 & 52 & 6.0 & 1 & 49 \\
\midrule
\multirow{5}{*}{Glove25} & \multirow{5}{*}{1,183,514} & \multirow{5}{*}{25} & 1.0000 & 50.41 & 50.0 & 1 & 141 & 43.0 & 1 & 422 \\
 &  &  & 0.9999 & 25.13 & 24.0 & 1 & 98 & 24.0 & 1 & 99 \\
 &  &  & 0.9995 & 16.02 & 15.0 & 1 & 87 & 15.0 & 1 & 78 \\
 &  &  & 0.9950 & 8.29 & 7.0 & 1 & 71 & 7.0 & 1 & 61 \\
 &  &  & 0.9500 & 4.23 & 4.0 & 1 & 44 & 4.0 & 1 & 49 \\
\midrule
\multirow{3}{*}{Yandex DEEP} & \multirow{3}{*}{1,000,000,000} & \multirow{3}{*}{96} & 1.0000 & 144.11 & 138.0 & 25 & 560 & - & - & - \\
 &  &  & 0.9950 & 16.30 & 16.0 & 4 & 115 & - & - & - \\
 &  &  & 0.9500 & 6.77 & 6.0 & 2 & 113 & - & - & - \\
\midrule
\multirow{3}{*}{BIGANN} & \multirow{3}{*}{1,000,000,000} & \multirow{3}{*}{128} & 1.0000 & 105.54 & 95.0 & 12 & 647 & - & - & - \\
 &  &  & 0.9950 & 11.81 & 11.0 & 2 & 117 & - & - & - \\
 &  &  & 0.9500 & 6.31 & 6.0 & 1 & 112 & - & - & - \\
\midrule
\multirow{3}{*}{\shortstack[l]{Facebook\\SimSearchNet++}} & \multirow{3}{*}{1,000,000,000} & \multirow{3}{*}{256} & 1.0000 & 2077.62 & 1577.0 & 44 & 61399 & - & - & - \\
 &  &  & 0.9950 & 8.43 & 7.0 & 1 & 133 & - & - & - \\
 &  &  & 0.9500 & 3.94 & 3.0 & 1 & 33 & - & - & - \\
\midrule
\multirow{3}{*}{Microsoft SPACEV1B} & \multirow{3}{*}{1,402,020,720} & \multirow{3}{*}{100} & 1.0000 & 554.80 & 500.0 & 85 & 2984 & - & - & - \\
 &  &  & 0.9950 & 22.06 & 19.0 & 3 & 141 & - & - & - \\
 &  &  & 0.9500 & 7.36 & 7.0 & 1 & 76 & - & - & - \\
\bottomrule
\end{tabular}
\smallskip
\caption{This table expands on \Cref{tab:coverage-stats}. For every dataset, this table describes the out-degree and in-degree statistics for a range of $\gamma$s.}
\label{tab:coverage-stats-full}
\end{table}

            \begin{table}[h]
            \centering
            % \small
            \setlength{\tabcolsep}{5pt}
            \begin{tabular}{lccccccccccccc}
            \toprule
            \multirow{3}{*}{\textbf{Dataset}} & \multirow{3}{*}{\textbf{Target Recall}} & \multicolumn{12}{c}{\textbf{Distance Computations}} \\
            \cmidrule(lr){3-14}
            & & \multicolumn{4}{c}{\textbf{$k = 1$}} & \multicolumn{4}{c}{\textbf{$k = 10$}} & \multicolumn{4}{c}{\textbf{$k = 100$}} \\
            \cmidrule(lr){3-6} \cmidrule(lr){7-10} \cmidrule(lr){11-14}
            & & \textbf{F} & \textbf{A} & \textbf{R} & \boldmath$\gamma$ & \textbf{F} & \textbf{A} & \textbf{R} & \boldmath$\gamma$ & \textbf{F} & \textbf{A} & \textbf{R} & \boldmath$\gamma$ \\
            \midrule
            \multirow{4}{*}{MNIST} & 0.90 & 254 & 175 & 0.69 & 0.99500 & 356 & 217 & 0.61 & 0.99600 & 1248 & 481 & 0.39 & 0.96500 \\
             & 0.95 & 335 & 220 & 0.66 & 0.99600 & 429 & 285 & 0.67 & 0.98500 & 1303 & 568 & 0.44 & 0.98500 \\
             & 0.97 & 380 & 236 & 0.62 & 0.99700 & 511 & 316 & 0.62 & 0.99500 & 1325 & 672 & 0.51 & 0.99500 \\
             & 0.99 & 517 & 350 & 0.68 & 0.99700 & 669 & 471 & 0.70 & 0.99500 & 1557 & 1073 & 0.69 & 0.98500 \\
            \midrule
            \multirow{4}{*}{Fashion-MNIST} & 0.90 & 266 & 184 & 0.69 & 0.99750 & 341 & 212 & 0.62 & 0.99550 & 1039 & 460 & 0.44 & 0.98000 \\
             & 0.95 & 330 & 216 & 0.66 & 0.99600 & 366 & 249 & 0.68 & 0.99900 & 1079 & 538 & 0.50 & 0.99550 \\
             & 0.97 & 357 & 244 & 0.68 & 0.99900 & 434 & 284 & 0.65 & 0.99550 & 1095 & 613 & 0.56 & 0.99800 \\
             & 0.99 & 454 & 287 & 0.63 & 0.99950 & 602 & 477 & 0.79 & 0.99970 & 1163 & 815 & 0.70 & 0.99900 \\
            \midrule
            \multirow{4}{*}{COCO-i2i} & 0.90 & 383 & 282 & 0.74 & 0.98000 & 562 & 333 & 0.59 & 0.99600 & 1966 & 649 & 0.33 & 0.96500 \\
             & 0.95 & 540 & 338 & 0.63 & 0.99600 & 772 & 435 & 0.56 & 0.99000 & 2054 & 872 & 0.42 & 0.98000 \\
             & 0.97 & 683 & 382 & 0.56 & 0.99800 & 876 & 567 & 0.65 & 0.98000 & 2089 & 1093 & 0.52 & 0.99700 \\
             & 0.99 & 902 & 495 & 0.55 & 0.99600 & 1406 & 857 & 0.61 & 0.99000 & 2534 & 1697 & 0.67 & 0.99000 \\
            \midrule
            \multirow{4}{*}{Glove25} & 0.90 & 1516 & 797 & 0.53 & 0.99990 & 1856 & 985 & 0.53 & 0.99980 & 4728 & 1917 & 0.41 & 0.99970 \\
             & 0.95 & 1977 & 1305 & 0.66 & 0.99995 & 2653 & 1492 & 0.56 & 0.99980 & 5487 & 3854 & 0.70 & 0.99995 \\
             & 0.97 & 2686 & 1526 & 0.57 & 0.99980 & 3103 & 2165 & 0.70 & 0.99980 & 7440 & 4357 & 0.59 & 0.99970 \\
             & 0.99 & 3396 & 3163 & 0.93 & 0.99995 & 4370 & 3519 & 0.81 & 0.99995 & 10656 & 6894 & 0.65 & 0.99995 \\
            \bottomrule
            \end{tabular}
            \smallskip
            \caption{This table shows the average distance computations to reach each target recall, broken out by Recall@$k$ ($k\in\{1,10,100\}$). Each block reports Fully-navigable (F, coverage $=1$), Almost-navigable (A, Pareto frontier over coverage $<1$), their ratio R\,$=$\,A\,/\,F, and $\gamma$ (coverage of the chosen almost-navigable point. For every target recall, $\gamma$-almost navigable graphs require fewer distance computations ($\sim 35 - 47$\% reduction) on average. The data points for this table were chosen by interpolating the curves in \Cref{fig:recall-v-dc}.}
            \label{tab:target-recall-dist}
            \end{table}
            
            \begin{table}
            \centering
            % \small
            \setlength{\tabcolsep}{5pt}
            \begin{tabular}{lccccccccccccc}
            \toprule
            \multirow{3}{*}{\textbf{Dataset}} & \multirow{3}{*}{\textbf{Target Recall}} & \multicolumn{12}{c}{\textbf{Average Degree}} \\
            \cmidrule(lr){3-14}
            & & \multicolumn{4}{c}{\textbf{$k = 1$}} & \multicolumn{4}{c}{\textbf{$k = 10$}} & \multicolumn{4}{c}{\textbf{$k = 100$}} \\
            \cmidrule(lr){3-6} \cmidrule(lr){7-10} \cmidrule(lr){11-14}
            & & \textbf{F} & \textbf{A} & \textbf{R} & \boldmath$\gamma$ & \textbf{F} & \textbf{A} & \textbf{R} & \boldmath$\gamma$ & \textbf{F} & \textbf{A} & \textbf{R} & \boldmath$\gamma$ \\
            \midrule
            \multirow{4}{*}{MNIST} & 0.90 & 19.64 & 10.32 & 0.53 & 0.99500 & 19.64 & 10.81 & 0.55 & 0.99600 & 19.64 & 6.42 & 0.33 & 0.96500 \\
             & 0.95 & 19.64 & 10.81 & 0.55 & 0.99600 & 19.64 & 8.04 & 0.41 & 0.98500 & 19.64 & 8.04 & 0.41 & 0.98500 \\
             & 0.97 & 19.64 & 11.45 & 0.58 & 0.99700 & 19.64 & 10.32 & 0.53 & 0.99500 & 19.64 & 10.32 & 0.53 & 0.99500 \\
             & 0.99 & 19.64 & 11.45 & 0.58 & 0.99700 & 19.64 & 10.32 & 0.53 & 0.99500 & 19.64 & 8.04 & 0.41 & 0.98500 \\
            \midrule
            \multirow{4}{*}{Fashion-MNIST} & 0.90 & 13.55 & 7.41 & 0.55 & 0.99750 & 13.55 & 6.71 & 0.50 & 0.99550 & 13.55 & 5.16 & 0.38 & 0.98000 \\
             & 0.95 & 13.55 & 6.85 & 0.51 & 0.99600 & 13.55 & 8.65 & 0.64 & 0.99900 & 13.55 & 6.71 & 0.50 & 0.99550 \\
             & 0.97 & 13.55 & 8.65 & 0.64 & 0.99900 & 13.55 & 6.71 & 0.50 & 0.99550 & 13.55 & 7.70 & 0.57 & 0.99800 \\
             & 0.99 & 13.55 & 9.72 & 0.72 & 0.99950 & 13.55 & 10.48 & 0.77 & 0.99970 & 13.55 & 8.65 & 0.64 & 0.99900 \\
            \midrule
            \multirow{4}{*}{COCO-i2i} & 0.90 & 28.59 & 9.15 & 0.32 & 0.98000 & 28.59 & 14.29 & 0.50 & 0.99600 & 28.59 & 7.60 & 0.27 & 0.96500 \\
             & 0.95 & 28.59 & 14.29 & 0.50 & 0.99600 & 28.59 & 11.27 & 0.39 & 0.99000 & 28.59 & 9.15 & 0.32 & 0.98000 \\
             & 0.97 & 28.59 & 16.66 & 0.58 & 0.99800 & 28.59 & 9.15 & 0.32 & 0.98000 & 28.59 & 15.27 & 0.53 & 0.99700 \\
             & 0.99 & 28.59 & 14.29 & 0.50 & 0.99600 & 28.59 & 11.27 & 0.39 & 0.99000 & 28.59 & 11.27 & 0.39 & 0.99000 \\
            \midrule
            \multirow{4}{*}{Glove25} & 0.90 & 50.41 & 25.13 & 0.50 & 0.99990 & 50.41 & 20.70 & 0.41 & 0.99980 & 50.41 & 18.47 & 0.37 & 0.99970 \\
             & 0.95 & 50.41 & 30.41 & 0.60 & 0.99995 & 50.41 & 20.70 & 0.41 & 0.99980 & 50.41 & 30.41 & 0.60 & 0.99995 \\
             & 0.97 & 50.41 & 20.70 & 0.41 & 0.99980 & 50.41 & 20.70 & 0.41 & 0.99980 & 50.41 & 18.47 & 0.37 & 0.99970 \\
             & 0.99 & 50.41 & 30.41 & 0.60 & 0.99995 & 50.41 & 30.41 & 0.60 & 0.99995 & 50.41 & 30.41 & 0.60 & 0.99995 \\
            \bottomrule
            \end{tabular}
            \smallskip
            \caption{This table shows the average degree to reach each target recall, broken out by Recall@$k$ ($k\in\{1,10,100\}$). Each block reports Fully-navigable (F, coverage $=1$), Almost-navigable (A, Pareto frontier over coverage $<1$), their ratio R\,$=$\,A\,/\,F, and $\gamma$ (coverage of the chosen almost-navigable point. For every target recall, $\gamma$-almost navigable graphs require fewer edges ($\sim 46 - 55$\% reduction) on average. The data points for this table were chosen by interpolating the curves in \Cref{fig:recall-v-dc}.}
            \label{tab:target-recall-deg}
            \end{table}

\end{document}

%% file: bib_macros.tex
  \usepackage{nth}
  \usepackage{intcalc}

  \newcommand{\cAAAI}[1]{AAAI\ Conference\ on\ Artificial (AAAI)}

%% file: math_macros.tex
\usepackage{amsmath,amsthm,amssymb}
\usepackage{mathtools}
\usepackage{dsfont}
\usepackage{bm}
\usepackage{xcolor}

\newcommand{\R}{\mathbb{R}}

\DeclareMathOperator*{\argmin}{arg\,min}

\newtheorem{theorem}{Theorem}
\newtheorem{lemma}[theorem]{Lemma}

\newtheorem{claim}[theorem]{Claim}

\theoremstyle{definition}
\newtheorem{definition}{Definition}

\theoremstyle{remark}

\Crefname{alg}{Algorithm}{Algorithms}